\documentclass{ws-ijqi}
\usepackage{amsmath,amssymb,dsfont,graphicx}

\newcommand{\Tr}{\operatorname{Tr}}

\newcommand{\hilb}[1]{\mathcal{#1}}

\newcommand{\map}[1]{\mathcal{#1}}

\newcommand{\linear}[1]{\boldsymbol{\mathsf{L}}(#1)}
\newcommand{\states}[1]{\boldsymbol{\mathsf{S}}(#1)}

\renewcommand{\ge}{\geqslant}

\def\sH{{\hilb{H}}}
\def\sK{{\hilb{K}}}
\def\>{\rangle}
\def\<{\langle}
\def\mL{{\map{L}}}
\def\mE{{\map{E}}}

\def\openone{\mathds{1}}

\newtheorem{thm}{Proposition}

\begin{document}

\markboth{F. Buscemi, M. Dall'Arno, M. Ozawa, V. Vedral}
{Universal Optimal Quantum Correlator}

\catchline{}{}{}{}{}

\title{Universal Optimal Quantum Correlator}

\author{Francesco Buscemi}

\address{Graduate School of Information Science, Nagoya
  University\\Chikusa-ku, Nagoya, 464-8601, Japan}

\author{Michele Dall'Arno}

\address{Centre for Quantum Technologies, National University of Singapore\\ 3
  Science Drive 2, 117543 Singapore, Republic of Singapore\\Graduate School of
  Information Science, Nagoya University\\Chikusa-ku, Nagoya, 464-8601, Japan}

\author{Masanao Ozawa}

\address{Graduate School of Information Science, Nagoya
  University\\Chikusa-ku, Nagoya, 464-8601, Japan}

\author{Vlatko Vedral}

\address{Atomic and Laser Physics, Clarendon Laboratory, University of
  Oxford\\Parks Road, Oxford OX13PU, United Kingdom\\Centre for Quantum
  Technologies, National University of Singapore\\ 3 Science Drive 2, 117543
  Singapore, Republic of Singapore}

\maketitle

\begin{history}
\received{Day Month Year}
\revised{Day Month Year}
\end{history}

\begin{abstract}
  Recently, a novel operational strategy to access quantum correlation functions
  of the form $\Tr[A \rho B]$ was provided in [F. Buscemi, M. Dall'Arno,
    M. Ozawa, and V. Vedral, arXiv:1312.4240]. Here we propose a realization
  scheme, that we call {\em partial expectation values}, implementing such
  strategy in terms of a unitary interaction with an ancillary system followed
  by the measurement of an observable on the ancilla. Our scheme is {\em
    universal}, being independent of $\rho$, $A$, and $B$, and it is optimal in
  a statistical sense. Our scheme is suitable for implementation with present
  quantum optical technology, and provides a new way to test uncertainty
  relations.
\end{abstract}


\section{Introduction}

Stochastic processes play a fundamental role in a plethora of fields such as
classical and quantum statistics~\cite{quant-stat},
thermodynamics~\cite{quant-term}, and field theory~\cite{quant-field}. They are
successfully described in terms of correlation functions, namely expectation
values of the product of dynamical variables. In classical theory, dynamical
variables are represented by real functions, while in quantum theory they are
represented by quantum observables -- i.e. Hermitian operators. Both theories
provide a prescription to directly measure the expectation value of any single
dynamical variable.

Classically, this prescription is sufficient to measure any correlation
function, since products of dynamical variables are dynamical variables
themselves. This is not the case in quantum theory, where the product of
non-commuting observables is not an observable in general. Thus, while formally
well-defined, quantum correlation functions appear to lack of a direct
operational interpretation.

Recently~\cite{BDOV13}, the present authors proposed a novel scheme -- referred
to as {\em ideal quantum correlator} -- which allows to operationally access the
expectation value of the product of any two observables $A$ and $B$ over any
quantum state $\rho$, namely any two-point quantum correlation function $\Tr[ A
  \rho B ]$. The scheme consists in a quantum preprocessing and classical
postprocessing strategy which is {\em universal}, being independent of $\rho, A,
B$, and is {\em optimal} in a statistical sense.

The aim of this work is to provide a simple realization scheme for our universal
optimal strategy, in terms of a unitary interaction $U$ with an ancillary system
followed by the measurement of an observable $Z$ on the ancilla. Our scheme is
universal, $U$ and $Z$ being fixed and independent of $\rho$, $A$, and $B$, and
optimal, minimizing the statistical error associated with the classical
postprocessing. Our scheme is suitable for implementation with present quantum
optical technology, and provides a new way to test uncertainty
relations~\cite{uncertainty}.

\section{Formalization}

Let us first fix the notation~\cite{wilde}. Let $\sH$ and $\sK$ be some Hilbert
spaces. We denote by $\linear{\sH,\sK}$ the set of all linear operators mapping
elements in $\sH$ to elements in $\sK$, with the convention that $\linear{\sH}
:= \linear{\sH,\sH}$. We denote by $\states{\sH}$ the set of all states, namely
all those operators $\rho \in \linear{\sH}$ such that $\rho \ge 0$ and
$\Tr[\rho]=1$. The identity matrix is denoted by the symbol $\openone$. Physical
transformation mapping quantum states on $\sH$ to quantum states on $\sK$ are
described by trace-preserving, completely positive linear maps $\map{M}:
\linear{\sH} \to \linear{\sK}$.

\section{Ideal Quantum Correlator}

Formally, for any Hilbert space $\hilb{H}$, the {\em ideal quantum correlator}
is defined~\cite{BDOV13} as the map $\map{T}: \linear{\hilb{H}} \to
\linear{\hilb{H}} \otimes \linear{\hilb{H}}$ such that
\begin{align}\label{eq:correlator}
  \Tr[\map{T}(\rho) \ (A \otimes B)] = \Tr[ A \rho B ],
\end{align}
for any observables $A, B \in \linear{\hilb{H}}$ and any state $\rho \in
\states{\hilb{H}}$.

The ideal quantum correlator $\map{T}$ is not a physical map, as it is not
Hermiticity preserving (HP). However, in Ref.~\cite{BDOV13} it was proved that
its expectation value can be decomposed in terms of physical -- namely,
completely positive (CP) and trace not-increasing -- maps. More precisely, it
was shown that any HP linear map $\mL: \linear{\sH}\to\linear{\sK}$ can be
decomposed as $\mL=\sum_i\lambda_i\mE_i$, where $\lambda_i$ are real
coefficients and $\mE_i:\linear{\sH}\to\linear{\sK}$ are completely positive
linear maps whose average, $\mE:=\sum_i\mE_i$, is trace-preserving. These are
called ``statistical decompositions'' of map $\mL$ and are illustrated in
Fig.~\ref{fig:stat-sim}. In Ref.~\cite{BDOV13}, the optimal decomposition of map
$\map{T}$ minimizing the statistical error associated with the postprocessing
was derived. In the next Section we will provide a realization scheme for such
decomposition.

\begin{figure}[htb]
  \begin{center}
    \includegraphics[width=.75\columnwidth]{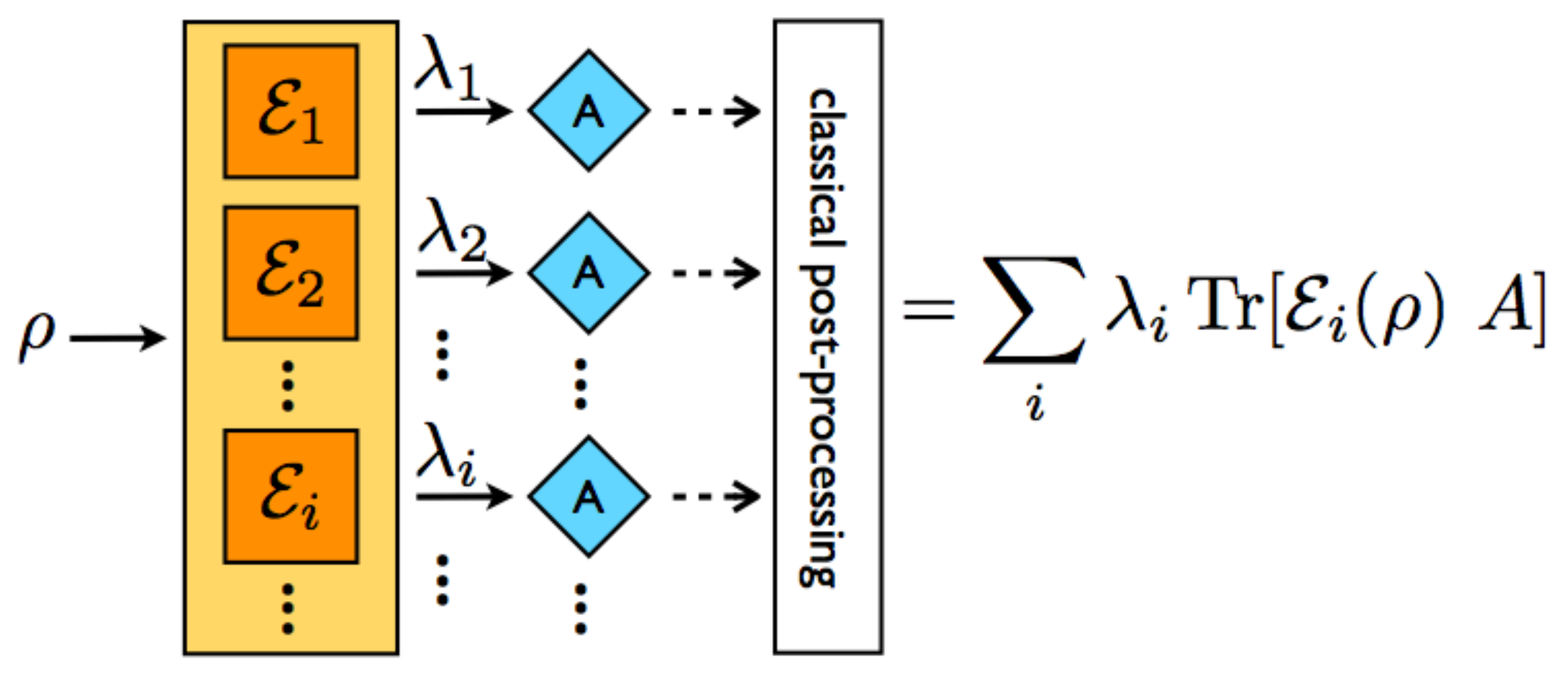}
  \end{center}
  \caption{Statistical decomposition of a non-physical transformation: (1)~the
    initial state $\rho$ goes through a quantum instrument, described by a
    collection of CP maps $\{\map{E}_i\}_i$; (2)~the outcome $i$, occurring with
    probability $p(i)=\Tr[{\map{E}}_i(\rho)]$, is recorded; (3)~the
    corresponding output state $\rho_i={\map{E}}_i(\rho)/p(i)$ is used to
    evaluate the expectation value $\langle A\rangle_i=\Tr[\rho_i\ A]$; (4)~all
    data are finally recombined as $\sum_i\lambda_ip(i)\langle
    A\rangle_i=\sum_i\lambda_i\Tr[\map{E}_i(\rho)\ A]$, for suitable real
    coefficients $\lambda_i$.}
 \label{fig:stat-sim}
\end{figure}

\section{Partial Expectation Values}

The following representation theorem provides a realization scheme to access
quantum correlation functions as in Eq.~\eqref{eq:correlator}, in terms of
\emph{partial expectation values}.

\begin{thm}[Partial Expectation Values]\label{prop:partexpvalues}
  For any linear HP map $\map{L}: \linear{\hilb{H}}\to\linear{\hilb{K}}$, there
  exists a finite dimensional ancillary quantum system $\sK'$, an isometry
  $V:\sH\to\sK\otimes\sK'$ and an observable $Z\in\linear{\sK'}$, such that
  \begin{align}
    \Tr[V\rho V^\dag \ (A\otimes Z)] = \Tr[\mL(\rho)\ A],
  \end{align}
  for all states $\rho\in\states{\sH}$ and all observables
  $A\in\linear{\sK}$. Equivalently,
  \begin{align}
    \mL(\rho)=\Tr_{\sK'}[V\rho V^\dag \ (\openone\otimes Z)],
  \end{align}
  namely, the action of $\mL$ can be written as a partial expectation value.
\end{thm}

\begin{proof}
  Let $\mL(\rho)=\sum_i\lambda_i\mE_i(\rho)$ be a statistical decomposition of
  $\mL$. Then, following Stinespring-Kraus's representation
  theorem~\cite{stinespring}, there exist $\sK'$ ancillary Hilbert space,
  $V:\sH\to\sK\otimes\sK'$ isometry, and $\{P^i\}_i$ POVM on $\sK'$ such that
  \begin{align*}       
    \mE_i(\rho)=\Tr_{\sK'}[V\rho V^\dag \ (\openone_{\sK}\otimes P^i_{\sK'})].
  \end{align*}
  The statement is recovered by setting $Z:=\sum_i\lambda_iP^i$.
\end{proof}

Notice that Proposition~\ref{prop:partexpvalues} can be regarded as a
generalization of Stinespring-Kraus's representation theorem~\cite{stinespring}
to arbitrary HP map. The idea of partial expectation values is depicted in
Fig.~\ref{fig:real} below.

\begin{figure}[htb]
  \begin{center}
    \includegraphics[width=.75\columnwidth]{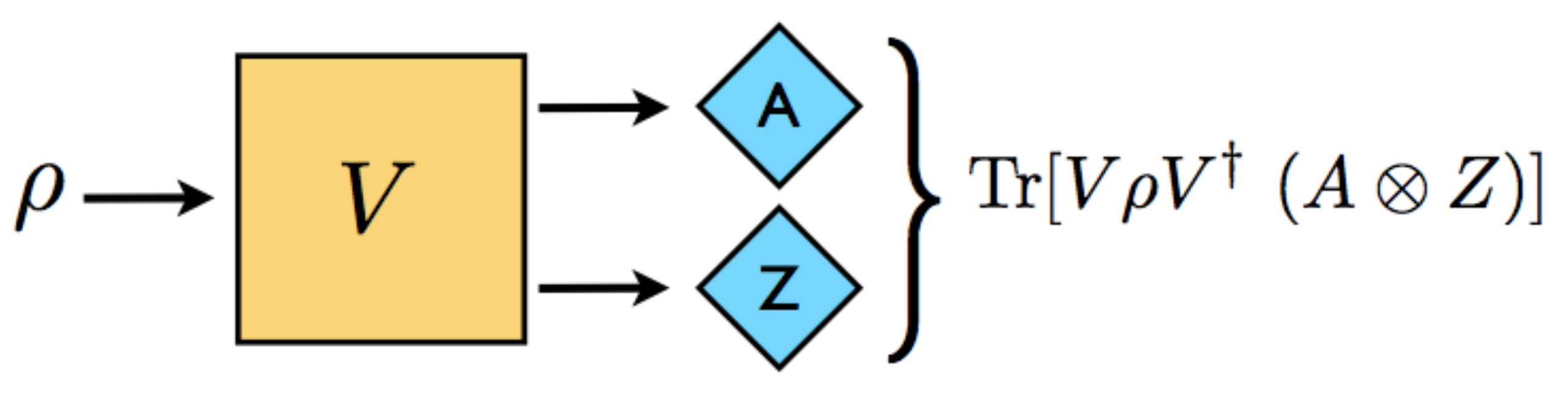}
  \end{center}
  \caption{Universal optimal strategy to access quantum correlation functions in
    terms of \emph{partial expectation values}, see
    Proposition~\ref{prop:partexpvalues}. The isometry $V$ and the ancillary
    observable $Z$ do not depend neither on the input state $\rho$ nor on the
    final observable $A$, but only on the linear HP map $\mL$. It holds that
    $\Tr[V\rho V^\dag\ (A\otimes Z)]=\Tr[\mL(\rho)\ A]$, for all input states
    $\rho$ and all final observables $A$.}
  \label{fig:real}
\end{figure}

\section{Conclusion}

We provided a simple realization scheme for accessing quantum correlations, in
terms of a unitary interaction $U$ with an ancillary system followed by the
measurement of an observable $Z$ on the ancilla. Our scheme is universal, $U$
and $Z$ being independent of $\rho$, $A$, and $B$, and optimal, minimizing the
statistical error associated with the classical postprocessing. Our scheme is
suitable for implementation with present quantum optical technology, and
represents a new way to test uncertainty relations~\cite{uncertainty}.

\section*{Acknowledgment}

This work was supported by the Ministry of Education (Singapore), the Ministry
of Manpower (Singapore), the National Research Foundation (Singapore), the EPSRC
(UK), the Templeton Foundation, the Leverhulme Trust, the Oxford Martin School,
the Oxford Fell Fund and the European Union, the JSPS (Japan society for the
Promotion of Science) Grant-in-Aid for JSPS Fellows No. 24-0219, and the JSPS
KAKENHI No. 26247016.

\end{document}